\documentclass[11pt]{article}
\usepackage[T1]{fontenc}
\usepackage[letterpaper,margin=1in]{geometry}
\usepackage{lmodern,microtype,amsmath,amssymb,amsthm,mathtools}
\usepackage{enumitem,needspace,booktabs,array}
\usepackage[hidelinks]{hyperref}
\hypersetup{pdftitle={Constant-Probability Witness Isolation\\ Implies $\mathrm{NP}\subseteq\mathrm{P/poly}$},pdfauthor={Sebastian Ben Daniel},pdfsubject={Tagged affine gadgets, constant-success collapse, and logarithmic regional thresholds},bookmarksnumbered=true}
\newtheorem{theorem}{Theorem}[section]
\newtheorem{lemma}[theorem]{Lemma}
\newtheorem{proposition}[theorem]{Proposition}
\newtheorem{corollary}[theorem]{Corollary}
\theoremstyle{definition}
\newtheorem{definition}[theorem]{Definition}
\theoremstyle{remark}
\newtheorem{remark}[theorem]{Remark}
\DeclareMathOperator{\Sol}{Sol}
\DeclareMathOperator{\poly}{poly}

\newcommand{\F}{\mathbb F}
\newcommand{\I}{\mathcal I}
\newcommand{\ind}{\mathbf 1}
\newcommand{\NP}{\mathrm{NP}}
\newcommand{\Ppoly}{\mathrm{P}/\mathrm{poly}}
\newcommand{\USAT}{\mathrm{USAT}}
\newcommand{\Yes}{\mathrm{Yes}}
\newcommand{\No}{\mathrm{No}}
\newcommand{\acc}{\mathrm{acc}}
\numberwithin{equation}{section}
\title{Constant-Probability Witness Isolation\\ Implies $\mathrm{NP}\subseteq\mathrm{P/poly}$}
\author{Sebastian Ben Daniel\\[2pt]\small Ben Gurion University}
\date{23 September 2026}
\begin{document}
\maketitle
\begin{abstract}
Valiant and Vazirani isolate a satisfying assignment of a circuit with
probability $\Omega(1/n)$. Dell, Kabanets, van Melkebeek, and Watanabe
showed that success above $2/3$ implies $\mathrm{NP}\subseteq\mathrm{P/poly}$
and asked about the range in between. We show that every positive constant
already implies the collapse: if a randomized nonuniform polynomial-size
pruning procedure succeeds with probability $\epsilon$ on affine circuit
inputs with at most $2^{\lfloor 2/\epsilon\rfloor}$ satisfying assignments,
then $\mathrm{NP}\subseteq\mathrm{P/poly}$. Success $10/\log L$ on affine
inputs with at most $L^{1/3}$ satisfying assignments suffices, where $L$ is
the description length, and on inputs with one or two satisfying
assignments the threshold $2/3$ drops to $3/5$. No cryptographic assumption
is used, and the procedure may read the entire circuit.

The proof compiles a pool of circuits into one circuit whose satisfying
assignments are indexed by tags in $\mathbb{F}_2^d$. Each member is assigned
an affine region of tag space, and if one member is unsatisfiable, the
satisfying set shrinks to that member's region. Because regions may overlap
and have different dimensions, the collapse reduces to a combinatorial
bound: no set of tags meets more than a $2/d$ fraction of an equally
weighted family of affine subspaces of all dimensions below $d$ in exactly
one point. This regional counting cannot go below order $1/\log L$. The
range between $\Theta(1/n)$, achieved by affine hashing, and $O(1/\log n)$
remains open.
\end{abstract}
\vspace{6pt}
\noindent\textbf{Keywords:} witness isolation; affine spaces; tagged gadgets;
nonuniform domination; selectors.

\section{Introduction}
\label{sec:intro}
A witness-isolation procedure takes a Boolean circuit $H$ and outputs a circuit
on the same variables whose satisfying set is a subset of $\Sol(H)$, with the
objective of leaving exactly one assignment. Its output is a circuit, not the
remaining witness. Valiant and Vazirani~\cite{VV86} achieve success
$\Omega(1/n)$ for $n$ witness variables. Dell, Kabanets, van Melkebeek, and
Watanabe (DKMW)~\cite{DKMW13} show that success $2/3+1/\poly(L)$ implies
$\NP\subseteq\Ppoly$, using success promises only on singleton and two-point
inputs. Their intermediate-success question asks how far this threshold can
be lowered.

We study the problem already under the promise that $\Sol(H)$ is a nonempty
affine subspace over $\F_2$. The circuit is the entire input: no origin, basis,
or dimension is supplied, and arbitrary use of its syntax is allowed.
Throughout, $L$ denotes description length and $n$ denotes witness length.
The following results are nonuniform complexity-collapse implications, not
unconditional separations of complexity classes.

\subsection{Every positive constant, already on bounded-witness inputs}
\begin{theorem}[Constant-witness obstruction to every constant guarantee]
\label{thm:constantintro}
Fix $0<\epsilon\le1$ and put $d=\lfloor2/\epsilon\rfloor+1$.
If a randomized nonuniform polynomial-size pruning transformation has success
at least $\epsilon$ on every nonempty affine input with at most $2^{d-1}$
satisfying assignments, at every sufficiently large description length, then
$\NP\subseteq\Ppoly$.
Consequently, for each fixed positive $\epsilon\le1$, an efficient nonuniform
pruning transformation with affine-input guarantee $\epsilon$ exists if and
only if $\NP\subseteq\Ppoly$.
\end{theorem}

The witness bound depends only on $\epsilon$, not on $L$ or $n$. The success
promise is used only on one-NO tagged gadgets, whose witness counts are powers
of two at most $2^{d-1}$. The all-YES gadgets have $2^d$ witnesses; they are
used for averaging but need only satisfy the all-seed pruning condition, not
any isolation guarantee. Theorem~\ref{thm:constants} and
Corollary~\ref{cor:constants} prove the statement.

\begin{theorem}[An inverse-logarithmic collapse threshold]
\label{thm:logintro}
If a randomized nonuniform polynomial-size pruning transformation isolates
every nonempty affine input with at most $L^{1/3}$ satisfying assignments with
probability at least $10/\log L$, at every sufficiently large description
length $L$, then $\NP\subseteq\Ppoly$.
\end{theorem}

Theorem~\ref{thm:logarithmic} proves this using a polynomial-size family of
regions. In particular, no constant-success guarantee remains outside the
collapse implication. Explicit small gadgets also improve the threshold to
$3/5$ with success required only on one- and two-witness inputs
(Proposition~\ref{prop:pointsandlines}); pruning is still required on other
inputs, including the four-witness all-YES gadgets.

\paragraph{Proof overview.}
Attach an affine region $K_i\subseteq\F_2^d$ to each member $E_i$ of a canonically
ordered pool. A guarded block uses the known dummy witness $0^\ell$ when a
region equation holds, and the unique witness of $E_i$ otherwise. If every
member is YES, the satisfying set is the graph of an affine map on all tags.
If $E_j$ is NO, the tag is forced into $K_j$, and all blocks of $E_j$ become
zero. The resulting region is therefore known from the other members'
witnesses. Tags remain distinct even when source witnesses coincide.

For a tag set $T$, let $f_{\mathcal F}(T)$ be the fraction of regions meeting
$T$ in exactly one point. All designations of one pool use the same compiled
description and the same fixed randomness list. Their average regional
singleton frequency is exactly the list average of $f_{\mathcal F}(T)$.
A uniform bound below the tester's threshold makes at least one designation
accept. Ko--Ogihara domination then produces polynomial advice.

Disjoint regions permit one retained point in every region and hence have
$\max_{|T|\ne1}f_{\mathcal F}(T)=1$. Their symmetric threshold tends to $1/2$.
For regions distributed uniformly across dimensions $0,\ldots,d-1$, the
affine-flag bound instead gives $\max_T f_{\mathcal F}(T)\le2/d$.
It is this simultaneous control of every retained tag set, not overlap by
itself, that removes the constant threshold.

\paragraph{The remaining range and the length convention.}
The theorems state description-length guarantees $p(L)$, with efficiency and
advice also measured in $L$. To compare with affine hashing, define
$p_{\mathrm{wit}}(n)$ to be a success lower bound on every affine circuit with
$n$ witness bits, retaining the same description-length efficiency convention.
Our compiled gadgets have $n\le L$ and $n\to\infty$, so
$p_{\mathrm{wit}}(n)\ge10/\log n$ implies the hypothesis used in
Theorem~\ref{thm:logintro}. Padding can increase $L$ without changing $n$;
these conventions are not otherwise interchangeable.
To our knowledge, the best general achievable scale remains $\Theta(1/n)$
from affine hashing~\cite{VV86}. The present results do not settle the
intermediate scales, in particular
$p_{\mathrm{wit}}(n)=\omega(1/n)$ with
$p_{\mathrm{wit}}(n)=o(1/\log n)$, under $\NP\not\subseteq\Ppoly$ alone.
They also do not exclude every positive leading constant in a
$1/\log L$ guarantee.

The companion report~\cite{BD26Crypto} studies cryptographic
security--isolation reductions, nearly optimal bounds for checkable filters,
and an advantage-sensitive oracle-output barrier.

\paragraph{Organization.}
Section~\ref{sec:prelim} gives the circuit model and elementary isolation facts.
Section~\ref{sec:pools} proves the tagged regional-collapse theorem, its explicit
families, and the multiscale consequences. Section~\ref{sec:sharpness} supplies
sharpness laws, their finite-randomness interpretation, and a floor for the
regional counting method. Section~\ref{sec:poolliterature} compares the proof
with DKMW and the selector literature.

\section{Preliminaries}
\label{sec:prelim}
All logarithms are base two. Efficient transformations are classical,
possibly nonuniform polynomial-size circuit families; their advice depends
only on the input-description length.

\subsection{Circuits and pruning}
We use a fixed bounded-fan-in Boolean circuit encoding. A length-$L$ description specifies at most $L$ input variables. Padding does not change the ordered variables or the function. For a circuit $H$, let
\[
 \Sol(H)=\{x\in\{0,1\}^{n(H)}:H(x)=1\}.
\]
The relation $H\equiv H'$ means that the circuits have the same ordered variables and compute the same Boolean function. Unless explicitly stated otherwise, circuit-description length, not gate count or witness length, is the parameter in the isolation guarantee.

\begin{definition}[Pruning transformation]
\label{def:pruning}
A randomized polynomial-size pruning transformation is a possibly nonuniform circuit family $\I=(\I_L)_L$ that maps an input description $\langle H\rangle$ and a uniform random string $\rho$ to the description of a Boolean circuit $\I(H;\rho)$ on the same ordered variables. Its size, random-seed length, and output-description length are bounded by fixed polynomials in $L$. For every input and every seed,
\[
 \Sol(\I(H;\rho))\subseteq\Sol(H).
\]
Write $D_\I(H;\rho)=\Sol(\I(H;\rho))$. An affine-input isolation guarantee $p(L)$ means that
\[
 \Pr_\rho[|D_\I(H;\rho)|=1]\ge p(L)
\]
for every length-$L$ circuit whose satisfying set is a nonempty affine subspace over $\F_2$.
\end{definition}

Pruning is required on every input, while the success promise concerns only affine inputs. The output's unsuccessful realizations, including the empty set and sets with several witnesses, remain part of the experiment. Success is not conditioned on finding or recognizing a unique witness. Every test below evaluates one sampled output circuit at several points without rerunning $\I$ between evaluations.

\subsection{The DKMW bound under the affine promise}
\begin{proposition}[DKMW already applies to affine inputs]
\label{prop:dkmwaffine}
If a randomized nonuniform polynomial-size pruning transformation isolates every nonempty affine-subspace input with probability at least $2/3+\varepsilon(L)$, where $\varepsilon(L)$ is inverse-polynomially bounded below at all sufficiently large lengths, then $\NP\subseteq\Ppoly$. Its success guarantee is needed only on inputs with one or two satisfying assignments.
\end{proposition}
\begin{proof}
Theorem~3.4 of DKMW~\cite{DKMW13} needs only two properties: a unique witness is retained with probability at least $p_1$, and on an input with exactly two witnesses the output fails to retain both with probability at least $p_2$. The condition $p_1+p_2/2\ge1+1/\poly(L)$ implies $\NP\subseteq\Ppoly$. A nonempty singleton and an arbitrary two-point binary set are affine spaces of dimensions zero and one. Pruning and successful isolation therefore give both properties with $p_1,p_2\ge2/3+\varepsilon(L)$, so $p_1+p_2/2\ge1+3\varepsilon(L)/2$. The finitely many exceptional lengths can be handled nonuniformly.
\end{proof}

\subsection{Unique-satisfiability transfer}
We include the standard promise-transfer argument so that the nonuniform
collapse proof is complete within this report.

\begin{lemma}[Unique-satisfiability transfer {\cite{VV86,DKMW13,BD26}}]
\label{lem:usatransfer}
Let $\USAT$ be the promise problem whose YES circuits have exactly one
satisfying assignment and whose NO circuits have none. If $\USAT$ has
nonuniform polynomial-size circuits, then $\NP\subseteq\Ppoly$.
A bounded-error randomized circuit family for the promise problem also suffices.
\end{lemma}
\begin{proof}
The polynomial-size transfer is the argument of DKMW~\cite[Lemma 3.2]{DKMW13}; we include the size accounting. For each length-$n$ SAT input, repeat the Valiant--Vazirani pruning reduction polynomially many times. Each satisfiable input has a uniquely satisfiable pruning with inverse-polynomial probability in one trial. Repeating enough times makes the probability of no such trial smaller than $2^{-n-1}$. A union bound over at most $2^n$ descriptions gives a fixed polynomial-size list of random strings that produces a uniquely satisfiable pruning for every satisfiable input of that length. Every pruning of an unsatisfiable input remains unsatisfiable.

Apply the promised decider to all prunings and take the OR. A satisfiable input has at least one promised YES pruning, so the OR is one, irrespective of answers on its unpromised prunings. An unsatisfiable input has only promised NO prunings, so the OR is zero. All pruned descriptions have length $n^{O(1)}$.

For a randomized promised decider, amplify its error on each promised input and fix its coins by a union bound over all descriptions of each relevant length. The amplification costs a polynomial factor. The resulting circuits have polynomial size. Polynomial-time many-one reductions then give the result for every fixed language in $\NP$.
\end{proof}

\subsection{The affine-hashing benchmark}
The following standard affine-target form of hashing records the achievable
inverse-witness-length guarantee when no target dimension is supplied.

\begin{lemma}[Affine hashing without the target dimension {\cite{VV86,BD26}}]
\label{lem:zerohash}
For every $0\le d\le n$, choose $i$ uniformly from $\{0,\ldots,d\}$ and
then a uniform linear map $h:\F_2^n\to\F_2^i$. The filter
$F_h(x)=\ind_{\{h(x)=0\}}$ is chosen independently of the target and satisfies
\[
 \Pr_{i,h}\bigl[|W\cap h^{-1}(0)|=1\bigr]\ge\frac{1}{4(d+1)}
\]
for every nonempty affine $W\subseteq\F_2^n$ of dimension at most $d$.
Conditional on $i=\dim W$, the restriction of $h$ to the direction space
of $W$ is invertible with probability at least $1/4$, and invertibility
makes the intersection a singleton. No basis or exact dimension of $W$
is supplied to the filter.
\end{lemma}
\begin{proof}
Choose $i$ uniformly from $\{0,\ldots,d\}$ and a uniform linear map $h:\F_2^n\to\F_2^i$. Accept $x$ if and only if $h(x)=0$. Let $W=v+U$ have dimension $m\le d$. Conditional on $i=m$, the restriction $h|_U$ is represented by a uniform $m\times m$ binary matrix. If it is invertible, the equation $h(v+u)=0$ has exactly one solution $u\in U$.

The invertibility probability is
\[
 \prod_{j=1}^m(1-2^{-j})\ge\frac14.
\]
For $m=0$, the empty product equals one. For $m\ge1$, factor out the first term $1/2$ and use
$\prod_{j=2}^m(1-2^{-j})\ge1-\sum_{j=2}^m2^{-j}\ge1/2$.
Multiplying by the probability $1/(d+1)$ that $i=m$ proves the claim.
\end{proof}

\paragraph{Finite randomness.}
Statements that choose uniformly from an arbitrary finite set describe an
ideal distribution unless its cardinality is a power of two. Such choices
admit bounded-fair-coin approximations of any prescribed inverse-polynomial
accuracy with polynomial overhead. In particular, the ideal hashing lemma
above supplies an $\Omega(1/n)$ bounded-coin guarantee after an error smaller
than a fixed fraction of its lower bound. All pruning statements hold on every
seed. Sharpness claims below distinguish exact ideal probabilities from their
finite-coin approximations.

\section{Tagged gadgets and regional counting}
\label{sec:pools}
The all-YES satisfying set of the gadget is indexed bijectively by a tag
$t\in\F_2^d$. A NO member forces that tag into its assigned affine region.
Regions may overlap and have different dimensions. The exact projection
identity below replaces the disjointness argument and removes any need to
separate colliding witnesses.

\subsection{Normalization and fixed lists}
Write $\USAT=(\Yes,\No)$ for the unique-versus-zero promise problem, and let
$\Yes_\ell,\No_\ell$ denote its length-$\ell$ descriptions. Normalize a
predicate $E$ of arity $n(E)\le\ell$ by
\[
 \widehat E(z)=E(z_1,\ldots,z_{n(E)})\wedge
       \bigwedge_{j=n(E)+1}^{\ell}(z_j=0).
\]
For $E\in\Yes_\ell$, write $w(E)\in\{0,1\}^{\ell}$ for its zero-extended
witness. We continue to sort and name the original descriptions; the
normalization need not preserve description length.
Let $Z_\ell(z)=\bigwedge_{j=1}^{\ell}(z_j=0)$ be the dummy predicate.

\begin{lemma}[One fixed list for all promised classes]
\label{lem:poollists}
Let $\I$ be a randomized nonuniform polynomial-size pruning transformation,
and let $\mathcal C_L$ be any class of length-$L$ inputs. Write
$q_H=\Pr_\rho[|\Sol(\I(H;\rho))|=1]$.
For every inverse-polynomial $\gamma>0$, there is a deterministic list-producing
family $B_L$ with $J=O((L+1)/\gamma^2)$ entries such that every entry prunes its
input and the fraction of isolating entries is at least $q_H-\gamma$ for every
$H\in\mathcal C_L$. One list of seeds works simultaneously on any union of
promised classes, even when their success lower bounds differ.
\end{lemma}
\begin{proof}
Use the fixed-list argument of DKMW~\cite[Lemma 2.2]{DKMW13}.
Choose $J$ independent uniform seeds of $\I_L$. For a fixed $H$, Hoeffding's
inequality bounds the probability of an empirical frequency below $q_H-\gamma$
by $\exp(-2J\gamma^2)$. There are at most $2^L$ descriptions in the entire
union of classes. With a sufficiently large absolute constant in $J$, the
union bound is less than one. Hardwire one such seed list and the advice for
$\I_L$. Every output prunes on every seed. The resulting list is a function
of the input description alone; no recognition of the promised classes or
of successful seeds is needed.
\end{proof}

\subsection{Region families and the tagged gadget}
\begin{definition}[Region family]
\label{def:regions}
A region family of tag dimension $d\ge1$ is a nonempty finite sequence
$\mathcal F=(K_1,\ldots,K_m)$ of nonempty affine subspaces of $\F_2^d$,
with repetitions allowed. Each region has a fixed presentation
\[
 K_i=\{t:\lambda_{i,a}(t)=\beta_{i,a}\quad(1\le a\le r_i)\},
 \qquad r_i=d-\dim K_i,
\]
by linearly independent linear forms. Put $s_{\max}=\max_i\dim K_i$.
For $T\subseteq\F_2^d$, define
\[
 f_{\mathcal F}(T)=\frac1m|\{i:|T\cap K_i|=1\}|,
 \quad g=\max_t f_{\mathcal F}(\{t\}),
 \quad c'=\max_{|T|\ne1}f_{\mathcal F}(T),
 \quad c=\max\{g,c'\}.
\]
\end{definition}

\begin{definition}[Tagged gadget]
\label{def:taggadget}
For a canonically ordered set $Q=\{E_1<\cdots<E_m\}$ of distinct length-$\ell$
descriptions, take tag variables $t\in\{0,1\}^d$ and blocks
$x^{i,a}\in\{0,1\}^{\ell}$ for $1\le i\le m$, $1\le a\le r_i$. Define
\begin{align}
 G_{i,a}(t,x^{i,a})
 &=\bigl([\lambda_{i,a}(t)=\beta_{i,a}]\wedge Z_\ell(x^{i,a})\bigr)
   \vee\bigl([\lambda_{i,a}(t)\ne\beta_{i,a}]
                    \wedge\widehat E_i(x^{i,a})\bigr),\label{eq:tagblock}\\
 H_{\mathcal F}(Q)&=\bigwedge_{i,a}G_{i,a}.\label{eq:taggadget}
\end{align}
The ordered variables include all $d$ tag bits, including any unused ones.
The compiler receives $\mathcal F,Q$, but no witnesses and no designation.
\end{definition}

Let $R=\sum_i r_i\le md$. The witness arity is $n=d+R\ell$.
The unpadded description length is
\[
 O\bigl((d+R(\ell+d))\log(d+R(\ell+d)+2)\bigr).
\]
For a fixed nontrivial family this is $O(\ell\log\ell)$, so one may pad to
$L(\ell)=\ell^2$ at sufficiently large $\ell$. The same padding is used for
all pools and designations at that source length, and adds no variables.

\begin{lemma}[Satisfying sets and computable designated regions]
\label{lem:tagsets}
Suppose all members other than $E_j$ are YES circuits. Define the affine map
$\chi_j$ blockwise by
\[
 \chi_j^{i,a}(t)=
 \begin{cases}
 (\lambda_{i,a}(t)\oplus\beta_{i,a})w(E_i),&i\ne j,\\
 0^\ell,&i=j,
 \end{cases}
 \qquad
 R_j(Q)=\{(t,\chi_j(t)):t\in K_j\}.
\]
This region is computable from $\mathcal F,j$, and the other members'
witnesses, without a witness for $E_j$.
\begin{enumerate}[label=(\alph*),leftmargin=2em]
\item If every member is YES, then
\[
 S_Q:=\Sol(H_{\mathcal F}(Q))=\{(t,\chi(t)):t\in\F_2^d\},
 \qquad \chi^{i,a}(t)=(\lambda_{i,a}(t)\oplus\beta_{i,a})w(E_i).
\]
It is an affine space of dimension $d$, and
$R_j(Q)=S_Q\cap\{t\in K_j\}$ for every $j$, regardless of witness collisions.
\item If $E_j$ is the only NO member, then
$\Sol(H_{\mathcal F}(Q))=R_j(Q)$, an affine space of dimension $\dim K_j$.
\end{enumerate}
\end{lemma}
\begin{proof}
Fix $t$. The two disjuncts in a block are mutually exclusive. If its region
equation holds, the block has exactly the completion $0^\ell$. If the equation
fails, its only completion is $w(E_i)$ when $E_i$ is YES, and it has no
completion when $E_i$ is NO. Thus every tag has at most one completion.
In the all-YES case every tag has the displayed completion $\chi(t)$.
Each coordinate of $\chi$ is affine in $t$, and retaining the tag coordinates
makes its graph injectively parametrized and of dimension $d$.
On $K_j$, all multipliers belonging to member $j$ vanish, so
$\chi(t)=\chi_j(t)$. If $E_j$ is NO, a completion exists precisely on $K_j$,
and it is $\chi_j(t)$. Its graph has dimension $\dim K_j$.
In particular, equality of witnesses, including $w(E_i)=0^\ell$, never merges
two satisfying assignments with different tags.
\end{proof}

\subsection{Regional testers and domination}
An advice set $X$ consists of $m-1$ distinct members of $\Yes_\ell$ and their
normalized witnesses. For a threshold $p'$, the deterministic tester $T_X$
operates as follows on a length-$\ell$ circuit $y$.
\begin{enumerate}[label=(\arabic*),leftmargin=2em]
\item If $y\in X$, accept.
\item Otherwise form the canonical $Q=X\cup\{y\}$ and let $j$ be the position
of $y$ in $Q$. Compute the common list
$B_{L(\ell)}(H_{\mathcal F}(Q))=(C'_1,\ldots,C'_J)$.
Using the advice witnesses, enumerate $R_j(Q)$ and put
\[
 \acc_X(y)=\frac1J|\{r:|\Sol(C'_r)\cap R_j(Q)|=1\}|.
\]
Accept precisely when $\acc_X(y)<p'$.
\end{enumerate}
An affine basis for $K_j$ is obtained from its equations by Gaussian
elimination; this basis is used only to enumerate the tester's region, and is
not supplied to $\I$. The tester has size
$\poly(\ell,m,d,L,J)2^{s_{\max}}$, including evaluation of output descriptions
of polynomial length. The comparison can hardwire the integer
$\lceil Jp'\rceil$, so no computability assumption on a real-valued threshold
is needed. Off the unique-versus-zero promise, its answer is unrestricted.

\begin{lemma}[Soundness]
\label{lem:poolsound}
If at least a $p'$ fraction of the common list isolates every one-NO gadget
input, then every tester $T_X$ rejects every $y\in\No_\ell$.
\end{lemma}
\begin{proof}
Step (1) cannot fire. In step (2), $y$ is the only NO member, so the entire
satisfying set is the nonempty region $R_j(Q)$ by Lemma~\ref{lem:tagsets}.
Each list entry prunes the input. Every isolating entry therefore has exactly
one point in this region, giving $\acc_X(y)\ge p'$.
\end{proof}

\begin{lemma}[Exact regional averaging]
\label{lem:regionalaverage}
For an all-YES pool $Q=\{E_1<\cdots<E_m\}$, let $\tau(t,x)=t$ and
$D_r=\Sol(C'_r)$. Then
\begin{equation}
 \frac1m\sum_{j=1}^m\acc_{Q\setminus\{E_j\}}(E_j)
 =\frac1J\sum_{r=1}^J f_{\mathcal F}(\tau(D_r)).
 \label{eq:regionalaverage}
\end{equation}
\end{lemma}
\begin{proof}
Every designation skips step (1) and reaches the same compiled description
$H_{\mathcal F}(Q)$ and the same list. By pruning, $D_r\subseteq S_Q$.
Lemma~\ref{lem:tagsets} makes $\tau:S_Q\to\F_2^d$ a bijection and identifies
$R_j(Q)$ with the preimage of $K_j$. Hence
$|D_r\cap R_j(Q)|=|\tau(D_r)\cap K_j|$.
Sum the singleton indicators over $j,r$. Overlap is counted with its actual
multiplicity by $f_{\mathcal F}$; no disjointness or independence between the
regions is assumed.
\end{proof}

\begin{lemma}[Regional domination]
\label{lem:pooldom}
Let $\sigma\in[0,1]$ and suppose at least a $\sigma$ fraction of the common
list isolates every all-YES gadget; $\sigma=0$ imposes no success requirement.
If
\begin{equation}
 \max\{g,\sigma g+(1-\sigma)c'\}<p',
 \label{eq:regionalcondition}
\end{equation}
then every $S\subseteq\Yes_\ell$ with $N:=|S|\ge2m$ has an advice set
$X\subseteq S$ of size $m-1$ whose tester accepts at least $N/(2m)$ members.
For $\sigma=0$, the sufficient condition is simply $c<p'$.
\end{lemma}
\begin{proof}
Fix an $m$-subset $Q$ and let $\sigma_Q\ge\sigma$ be its list's global
singleton fraction. Singleton outputs project to singleton tag sets and
contribute at most $g$ to~\eqref{eq:regionalaverage}; other outputs contribute
at most $c'$. Its right-hand side is at most
$\sigma_Qg+(1-\sigma_Q)c'$. The maximum of this affine function on
$\sigma_Q\in[\sigma,1]$ is the left-hand side of
\eqref{eq:regionalcondition}. Thus some designation is accepted.

Choose $Q$ uniformly among all $m$-subsets of $S$, and then $y$ uniformly in
$Q$. The acceptance probability is at least $1/m$.
The map $(Q,y)\mapsto(Q\setminus\{y\},y)$ is bijective, and
\[
 \binom Nm m=\binom N{m-1}(N-m+1).
\]
It therefore induces a uniform $(m-1)$-subset $X$ followed by a uniform
$y\in S\setminus X$. Averaging over $X$, some fixed advice set accepts at
least $(N-m+1)/m\ge N/(2m)$ members. There is no conditioning on a good-pool
event and no collision loss.
\end{proof}

\subsection{The regional collapse theorem}
\begin{theorem}[Regional collapse]
\label{thm:regional}
Let $\mathcal F_\ell$ be region families with tag dimension, number of regions,
and $2^{s_{\max}}$ polynomially bounded in $\ell$, with their presentations
supplied as advice. Let $L(\ell)\ge\ell$ be a common polynomial padding length.
Suppose a randomized nonuniform polynomial-size pruning transformation has
success at least $p_{\No}(\ell)\in[0,1]$ on every one-NO gadget and at least
$p_{\Yes}(\ell)\in[0,1]$ on every all-YES gadget, at length $L(\ell)$.
If, eventually,
\begin{equation}
 p_{\No}-\max\{g_\ell,p_{\Yes}g_\ell+(1-p_{\Yes})c'_\ell\}
 \ge\eta(\ell)\ge1/\poly(\ell),
 \label{eq:regionalcollapse}
\end{equation}
then $\NP\subseteq\Ppoly$. In particular, with $p_{\Yes}=0$, success
$c(\mathcal F_\ell)+\eta$ on the one-NO class suffices. Pruning is still
required on all inputs and all seeds.
\end{theorem}
\begin{proof}
Put $\gamma=\eta/3$. Apply Lemma~\ref{lem:poollists} once to the union of both
gadget classes, obtaining one common seed list. Set
$p'=p_{\No}-\gamma$ and $\sigma=\max\{p_{\Yes}-\gamma,0\}$.
Changing $p_{\Yes}$ to $\sigma$ changes the affine expression in
\eqref{eq:regionalcollapse} by at most $\gamma|c'-g|\le\gamma$; its maximum
with $g$ can increase by no more. Therefore
\[
 \max\{g,\sigma g+(1-\sigma)c'\}
 \le p_{\No}-\eta+\gamma<p_{\No}-\gamma=p'.
\]
The list gives soundness and the hypotheses of Lemma~\ref{lem:pooldom}.

Start with $S_0=\Yes_\ell$. While $|S_j|\ge2m$, choose an advice set supplied
by that lemma and delete every input accepted by its tester. Each round
removes a fraction at least $1/(2m)$. Since $|\Yes_\ell|\le2^\ell$,
$\lceil2m\ell\rceil+1$ rounds suffice. Hardwire the advice sets, their
witnesses, and the residual set of fewer than $2m$ descriptions. Accept an
input if it is residual or any tester accepts. Lemma~\ref{lem:poolsound}
gives soundness, and deletion gives completeness.
There are $O(m\ell)$ testers with $O(m\ell)$ witness-advice bits each,
$J=O((L+1)/\gamma^2)$ list entries, and at most $2^{s_{\max}}$ evaluations
per entry. Every quantity is polynomial in $\ell$.
Thus $\USAT$ has polynomial-size circuits, and Lemma~\ref{lem:usatransfer}
yields the collapse. Finitely many source lengths are handled by advice.
\end{proof}

\begin{remark}[Disjoint regions recover the old pool threshold]
\label{rem:disjoint}
For $m\ge2$ pairwise disjoint nonempty regions, $g=1/m$ and $c'=1$: choose
one point from each region. With the same success lower bound $p$ on both
classes, the strict condition is $p>m/(2m-1)$, recovering
$(k+1)/(2k+1)$ at $m=k+1$. The original all-pairs pool realizes these regional
quantities. Its one-half limit is thus a limitation of the disjoint-region
count, not an obstacle to overlapping multiscale families.
Overlap alone is not sufficient: identical regions, for example, still have
$c=1$.
For any family with $1+c'-g>0$, the symmetric scalar threshold is
\[
 \rho(\mathcal F)=\max\left\{g,\frac{c'}{1+c'-g}\right\}.
\]
The degenerate case $g=1,c'=0$ gives no strict threshold in $[0,1]$.
\end{remark}

\subsection{Explicit families}
\begin{proposition}[All affine lines]
\label{prop:lines}
Let $d\ge2$, $N=2^d$, and let $\mathcal F$ contain each affine line of
$\F_2^d$ once. Then
\[
 f_{\mathcal F}(T)=\frac{|T|(N-|T|)}{\binom N2},\qquad
 g=\frac2N,\qquad c'=\frac{N}{2(N-1)},
\]
and the symmetric threshold is
\[
 \rho_N=\frac{N^2}{3N^2-6N+4}\le\frac13+\frac4{3N}.
\]
Success $\rho_N+\eta$ on all affine inputs with exactly $2$ or $N$ witnesses,
for an inverse-polynomially bounded-below margin $\eta$, implies
$\NP\subseteq\Ppoly$. More generally it suffices that the respective
success bounds satisfy
\[
 p_2+\frac{(N-2)^2}{2N(N-1)}p_N
 \ge\frac{N}{2(N-1)}+\eta.
\]
These hypotheses are understood at all sufficiently large description lengths.
\end{proposition}
\begin{proof}
Over $\F_2$ every two distinct points form a line. A line meets $T$ once
exactly when its endpoints lie on opposite sides of $T$. This proves the
formula for $f$. For singletons it is $2/N$. For nonsingletons it is maximized
at $|T|=N/2\ge2$, giving $c'$. Since
$c'-g=(N-2)^2/(2N(N-1))\ge0$, the maximum in
\eqref{eq:regionalcollapse} is $p_Ng+(1-p_N)c'$. The displayed asymmetric
condition therefore implies the theorem. Setting $p_2=p_N=p$ gives
$p>c'/(1+c'-g)=\rho_N$; the inverse-polynomial excess supplies its margin.
The algebraic identity
\[
 \rho_N-\frac13=\frac{6N-4}{3(3N^2-6N+4)}
\]
and $3N^2-6N+4\ge3N^2/2$ for $N\ge4$ yield the upper bound.
For fixed $d$, the family size and enumeration costs are constants apart
from polynomial factors in the source length, as required by
Theorem~\ref{thm:regional}.
\end{proof}

For $N=4,8,16$, the thresholds are $4/7$, $16/37$, and $64/169$.
In particular, every fixed guarantee strictly above $1/3$ is obstructed on
an appropriate constant-witness affine class.

\Needspace{11\baselineskip}
\begin{proposition}[Points and lines of $\F_2^2$]
\label{prop:pointsandlines}
Let $\mathcal F$ contain the four points and six lines of $\F_2^2$, once each.
Then
\[
 f_{\mathcal F}(T)=\frac{|T|(5-|T|)}{10},\qquad
 c=c'=\frac35,\qquad g=\frac25.
\]
Success $3/5+\eta$ on all affine inputs with one or two witnesses implies
$\NP\subseteq\Ppoly$. If success is also required on four-witness inputs,
a symmetric guarantee $1/2+\eta$ suffices. In both cases the margin is
inverse-polynomially bounded below at all sufficiently large lengths.
\end{proposition}
\begin{proof}
A set of $u$ tags contains $u$ point regions and meets $u(4-u)$ lines once.
The values of $f$ at $u=0,1,2,3,4$ are respectively
$0,2/5,3/5,3/5,2/5$. One-NO gadgets have one or two witnesses.
Use $p_{\Yes}=0$ in Theorem~\ref{thm:regional} for the first claim.
For the second, $c'/(1+c'-g)=(3/5)/(6/5)=1/2>g$.
No isolation promise is used on the four-witness all-YES gadgets in the first
claim, but their outputs must still prune those gadgets.
\end{proof}

\begin{lemma}[Multiscale affine-flag sum]
\label{lem:scalesum}
For each $0\le s<d$, let $K_s$ be a uniformly random affine $s$-subspace of
$\F_2^d$. For every fixed $T\subseteq\F_2^d$,
\begin{equation}
 \sum_{s=0}^{d-1}\Pr[|T\cap K_s|=1]\le2-2^{1-d}\le2.
 \label{eq:scalesum}
\end{equation}
Only the uniform marginal laws are needed.
\end{lemma}
\begin{proof}
This is the affine-flag argument used in~\cite{BD26}; we give the short
specialization needed here. Choose $K_{d-1}$ uniformly and successively choose
a uniform affine hyperplane of the preceding space, down to $K_0$.
Affine symmetry gives the required uniform marginals. If a point lies in an
affine space of positive dimension, a uniform affine hyperplane contains it
with probability $1/2$, since exactly one of the two translates of every
codimension-one direction space contains that point.

Along the descending chain, intersections with $T$ can only shrink.
Condition on the first singleton intersection occurring at dimension $j$.
Every subsequent step preserves its one point with probability $1/2$; after
it is lost, no later intersection can become nonempty. The conditional
expected number of singleton levels is therefore
$1+1/2+\cdots+2^{-j}\le2-2^{1-d}$.
If no singleton is reached the count is zero. Averaging proves the bound for
the chain, and linearity of expectation transfers it to any coupling with
the same marginals.
\end{proof}

Let $A_s$ be the number of affine $s$-subspaces of $\F_2^d$, and put
$M=\operatorname{lcm}(A_0,\ldots,A_{d-1})$.
Define $\mathcal F_d$ by listing every $s$-flat $M/A_s$ times.
Each dimension has total weight $1/d$, so
\begin{equation}
 c(\mathcal F_d)=\max_T\frac1d\sum_{s=0}^{d-1}
                  \Pr[|T\cap K_s|=1]\le\frac2d.
 \label{eq:multiscale}
\end{equation}
For fixed $d$, the family may be very large but is a constant-size object
independent of $\ell$. Its regions have at most $2^{d-1}$ points.

\begin{theorem}[Every positive constant]
\label{thm:constants}
Let $0<\epsilon\le1$ and $d=\lfloor2/\epsilon\rfloor+1$.
If a randomized nonuniform polynomial-size pruning transformation isolates
every affine input with at most $2^{d-1}$ satisfying assignments with
probability at least $\epsilon$, at every sufficiently large description
length, then $\NP\subseteq\Ppoly$. The success guarantee is used only on the
one-NO gadgets of $\mathcal F_d$.
\end{theorem}
\begin{proof}
The one-NO inputs have $2^s$ witnesses for $0\le s<d$. The fixed family's
parameters satisfy all resource requirements of Theorem~\ref{thm:regional},
and $L(\ell)=\ell^2$ is valid eventually. Use $p_{\Yes}=0$,
$p_{\No}=\epsilon$, and the positive constant margin
$\eta=\epsilon-2/d$ in that theorem.
\end{proof}

\begin{corollary}[Constant-success equivalence]
\label{cor:constants}
For every fixed $0<\epsilon\le1$, an efficient nonuniform pruning
transformation with affine-input guarantee $\epsilon$ exists if and only if
$\NP\subseteq\Ppoly$.
\end{corollary}
\begin{proof}
The forward implication is Theorem~\ref{thm:constants}. Conversely,
polynomial-size SAT circuits support prefix self-reduction: first test
satisfiability, then set each witness bit to zero whenever a satisfying
extension exists and to one otherwise. Output the singleton predicate of the
lexicographically first witness, or the empty predicate on an unsatisfiable
input. This has success one and prunes on every seed.
\end{proof}

\begin{theorem}[Logarithmic threshold]
\label{thm:logarithmic}
If a randomized nonuniform polynomial-size pruning transformation isolates
every affine input with at most $L^{1/3}$ satisfying assignments with
probability at least $10/\log L$, at every sufficiently large description
length $L$, then $\NP\subseteq\Ppoly$.
\end{theorem}
\begin{proof}
Set $d=\lfloor\log\ell\rfloor$, $N=2^d\le\ell$, and $m=Nd^2$.
Sample $m$ independent regions by first choosing $s$ uniformly in
$\{0,\ldots,d-1\}$ and then a uniform affine $s$-subspace.
For any fixed $T\subseteq\F_2^d$, the indicators of singleton intersection
are independent across the sampled regions, with mean at most $2/d$ by
Lemma~\ref{lem:scalesum}. Hoeffding's inequality bounds the probability of
an empirical frequency above $3/d$ by
$\exp(-2m/d^2)=e^{-2N}$. A union bound over all $2^N$ possible retained tag
sets is less than one. Thus there exists one family $\mathcal F_\ell$ with
$c(\mathcal F_\ell)\le3/d$ simultaneously for every $T$. Fix it as advice;
no efficient sampling or recognition of a good family is asserted or needed.

The family has at most $md=Nd^3$ guarded blocks. Its unpadded length is
$O(\ell^2\log^4\ell)$, so pad to $L(\ell)=\ell^3$ at sufficiently large
source lengths. Every designated region has at most $2^{d-1}\le\ell/2$
points, at most $L^{1/3}$, and can be enumerated in polynomial time.
The family description, including all linear equations, has
$O(md^2)=\poly(\ell)$ bits.
Put $x=\log L$. Since $d\ge x/3-1$, we have
\[
 c(\mathcal F_\ell)\le\frac9{x-3},\qquad
 \frac{10}{x}-\frac9{x-3}
 =\frac{x-30}{x(x-3)}\ge\frac1{2x}
 \quad(x\ge57).
\]
Apply Theorem~\ref{thm:regional} with $p_{\Yes}=0$ and
$\eta=1/(2\log L)$. All its list, advice, and enumeration costs are polynomial
in $\ell$.
\end{proof}

The same constructions with the weaker second-moment scale estimate $5$
give $5/d$ and $19/\log L$, respectively. Those bounds remain valid;
Lemma~\ref{lem:scalesum} improves their constants without changing the gadget.
The constant $10$ is not claimed optimal.
The witness-length implication $p_{\mathrm{wit}}(n)\ge10/\log n$ follows
because the constructed arities satisfy $n\le L$ and tend to infinity.

\subsection{Exactly checked small dimensions}
For dimension weights $a_s\ge0$ with $\sum_s a_s=1$, give each affine
$s$-flat weight $a_s/A_s$. The quantities $c,g,c'$ then use this weighted
average. Rational weights are realized exactly by repetitions; they need
not be rounded. In Table~\ref{tab:small}, the $c$ columns require success
only on the one-NO class, while the $\rho$ columns also use success on the
all-YES class. The parenthesized witness count is needed only for the latter.
\begin{table}[htbp]
\centering\small
\renewcommand{\arraystretch}{1.2}
\begin{tabular}{@{}ccllll@{}}
\toprule
$d$ & Witness counts & $c$, uniform & $c$, optimized & $\rho$, uniform & $\rho$, optimized\\
\midrule
2 & $1,2\;(4)$ & $5/8$ & $3/5$ & $1/2$ & $1/2$\\
3 & $1,2,4\;(8)$ & $19/42$ & $96/229$ & $76/195$ & $72/191$\\
4 & $1,2,4,8\;(16)$ & $589/1680$ & $219/674$ & $2356/7501$ & $280/929$\\
\bottomrule
\end{tabular}
\caption{Exact regional thresholds for flats of dimensions $0,\ldots,d-1$.
Uniform means $a_s=1/d$; each optimized column optimizes its own dimension
weights.}
\label{tab:small}
\end{table}

These values were independently checked by enumerating all $2^{2^d}$ retained
tag sets. For $d=2,3,4$, the respective counts of flats in the included
dimensions are $(4,6)$, $(8,28,14)$, and $(16,120,140,30)$.
The optimized values have rational primal and dual certificates, included
with the accompanying verification code. They are not merely floating-point
linear-program outputs. None of the asymptotic theorems depends on these
finite computations.

\section{Sharpness and the limit of regional counting}
\label{sec:sharpness}
\begin{proposition}[Tag-restriction transformations]
\label{prop:tagrestrictions}
Fix a family $\mathcal F$ with $d=O(\log L)$ and a polynomial-size sampler
for a distribution $\mu$ on subsets of $\F_2^d$, represented by truth tables.
Apply the filter $[t\in T]$, for $T\leftarrow\mu$, to the first $d$ variables
of an input circuit, returning the input unchanged if it has fewer than $d$
variables. This is a polynomial-size pruning transformation on every input.
On a canonical one-NO gadget with NO member in position $j$, its success is
$\Pr_{\mu}[|T\cap K_j|=1]$; on an all-YES gadget, it is
$\Pr_{\mu}[|T|=1]$.
\end{proposition}
\begin{proof}
Conjoining a predicate always prunes. A truth table on $d$ bits gives a
circuit of size $O(d2^d)$, which is polynomial in $L$.
On the gadget classes, the first $d$ variables are the tag, and
Lemma~\ref{lem:tagsets} makes the tag projection bijective on the full
satisfying set. This gives the two probabilities. Recognition of a syntactic
gadget is unnecessary; the success guarantees, but not pruning, are restricted
to the gadget classes.
\end{proof}

\begin{corollary}[Sharpness laws on the canonical gadget classes]
\label{cor:tagsharpness}
The following ideal tag laws attain the indicated regional thresholds.
\begin{enumerate}[label=(\roman*),leftmargin=2em]
\item For the all-lines family of Proposition~\ref{prop:lines}, choose a
uniform point with probability $\rho_N$ and a uniform affine hyperplane
otherwise. Success equals $\rho_N$ on every one-NO and every all-YES gadget.
\item For the points-and-lines family of Proposition~\ref{prop:pointsandlines},
choose a uniform $2$-subset with probability $3/5$ and a uniform $3$-subset
otherwise. Success equals $3/5$ on every one-NO gadget.
\item For the same family, choosing a uniform singleton or a uniform
$3$-subset, each with probability $1/2$, gives success exactly $1/2$ on every
one-NO and every all-YES gadget.
\end{enumerate}
The laws in (i) and (ii) have bounded-fair-coin implementations with success
at least the displayed threshold minus any prescribed inverse-polynomial
$\delta(L)>0$, and polynomial overhead. Law (iii) is exactly implementable
with bounded fair coins.
No claim of these guarantees on every affine circuit with the same witness
count is made.
\end{corollary}
\begin{proof}
For (i), a uniform affine hyperplane separates the endpoints $a,b$ of a fixed
line exactly when its nonzero linear form takes value one on $a\oplus b$.
This has probability $N/(2(N-1))=c'$. The hyperplane has $N/2\ge2$ points,
so it does not isolate an all-YES gadget. A uniform point lies on the line
with probability $2/N=g$ and always isolates the all-YES gadget.
Thus the line success is $(1-\rho_N)c'+\rho_Ng=\rho_N$, as is all-YES success.

For (ii), a fixed point lies in a uniform $2$-subset with probability $1/2$
and in a uniform $3$-subset with probability $3/4$. A fixed line meets these
sets once with probabilities $2/3$ and $1/2$, respectively. Both weighted
averages are $3/5$.
For (iii), point success is $(1/4+3/4)/2=1/2$, line success is
$(1/2+1/2)/2=1/2$, and the full tag set is isolated precisely on the
singleton branch. One branch bit and two bits selecting a tag implement this
law exactly, using the selected tag either as the singleton or as the point
omitted from the $3$-subset.

The ideal mixtures in (i) and (ii) use non-dyadic probabilities.
To implement them in the finite-coin model, approximate each finite choice
by bounded rejection sampling, with total variation error at most $\delta$.
For example, drawing a uniform integer from a finite interval succeeds per
trial with probability greater than $1/2$; $O(\log(1/\delta))$ trials make
fallback mass at most $\delta$, after allocating the error budget among the
constant number of choices. The intervals and outcome descriptions have
polynomial bit length, and the tag truth tables have polynomial length.
Total variation bounds every success-probability change, and conjoining the
sampled tag filter preserves pruning on all seeds, including fallback seeds.
Thus the ideal thresholds are approached to arbitrary inverse-polynomial
accuracy. Exact finite-coin success probabilities are dyadic, so non-dyadic
values such as $3/5$ are not asserted as exact circuit probabilities.
\end{proof}

\paragraph{The original disjoint-region sharpness example.}
For $m$ disjoint nonempty regions, choose a point $z_i$ in each. With
probability $s$, retain a uniform singleton from these $m$ points; otherwise
retain all $m$ points. The global singleton fraction is $s$ and every
regional singleton probability is $1-s+s/m$. At $s=m/(2m-1)$ these agree,
so the strict threshold test cannot use this counting inequality to go below
that value. This is an exact statement about subset laws. Its finite-coin
interpretation follows the same approximation convention as above. The
multiscale families evade the example because their regions are not disjoint
and their $c$ is uniformly small.

\begin{lemma}[A floor for regional counting]
\label{lem:regionalfloor}
For a nonempty family with $s_{\max}<d$,
\[
 c(\mathcal F)\ge\frac1{4(s_{\max}+1)},\qquad
 \rho(\mathcal F):=\max\left\{g,\frac{c'}{1+c'-g}\right\}
 \ge\frac1{4(s_{\max}+2)}.
\]
\end{lemma}
\begin{proof}
For the first bound, choose $i$ uniformly from $\{0,\ldots,s_{\max}\}$,
choose a uniform linear map $h:\F_2^d\to\F_2^i$, and set $T=h^{-1}(0)$.
For each fixed $K_j$, conditioning on $i=\dim K_j$ gives a singleton
intersection with probability at least $1/4$, by Lemma~\ref{lem:zerohash}.
Hence $\mathbb E f_{\mathcal F}(T)\ge1/[4(s_{\max}+1)]$, and the maximum
is at least this expectation.

For the second bound, choose $i$ uniformly from the set
$\{0,\ldots,s_{\max}\}\cup\{d\}$, which has $s_{\max}+2$ members.
With $q=1/[4(s_{\max}+2)]$, the same argument gives
$\mathbb E f_{\mathcal F}(T)\ge q$, and the extra choice $i=d$ gives
$a:=\Pr[|T|=1]\ge q$ by invertibility. The denominator $1+c'-g$ is positive:
a nonempty proper region has one point inside and one outside, whose pair
shows $c'>0$, while $g\le1$.
If $\rho(\mathcal F)<q$, then both $g<q$ and
$qg+(1-q)c'<q$. On the other hand, partitioning the sampled sets into
singletons and nonsingletons yields
\[
 \mathbb E f_{\mathcal F}(T)
 \le ag+(1-a)c'
 \le\max\{g,qg+(1-q)c'\}<q,
\]
contradicting the lower bound. The ideal finite choices here are used only in
an extremal averaging argument, not as exact bounded-coin algorithms.
\end{proof}

Since the tester enumerates a region, polynomial cost requires
$s_{\max}=O(\log L)$. The lemma places an $\Omega(1/\log L)$ floor on the
scalar thresholds supplied by this regional count, either with no all-YES
success promise or with the same scalar guarantee on both classes.
Theorem~\ref{thm:logarithmic} matches that order. The constants hidden in this
method comparison can depend on the polynomial resource bound.
This does not assert the same floor for asymmetric promises that impose a
much larger all-YES success probability.

\begin{remark}[What this method does not settle]
\label{rem:remaining}
In the witness-length convention, affine hashing achieves the $\Theta(1/n)$
scale, whereas $p_{\mathrm{wit}}(n)\ge10/\log n$ implies the collapse.
The results here do not settle the intervening scales or all inverse-logarithmic
leading constants. Improving this specific regional-counting argument to
$o(1/\log L)$ would require abandoning at least one ingredient of its stated
framework, such as explicit enumeration of the tested region or the scalar
extremal bound. The lemma does not prove that every possible complexity
reduction must use high-dimensional inputs, nor does it rule out other
gadgets, statistics, or advice constructions.
\end{remark}

\section{Relation to selectors, membership comparability, and DKMW}
\label{sec:poolliterature}
DKMW's two-circuit comparison~\cite[Theorems 3.3--3.4]{DKMW13} yields a
collapse from success above $2/3$ on singleton and two-point inputs.
The earlier all-pairs pool extends its disjoint-region count to
$(k+1)/(2k+1)$. The tagged construction does not simply sharpen that count:
it changes the region geometry while preserving one common canonical
presentation over every designation. Theorem~\ref{thm:constants} addresses
every positive constant in DKMW's intermediate range, and
Theorem~\ref{thm:logarithmic} gives an inverse-logarithmic threshold without a
cryptographic assumption. The success promise is restricted to affine inputs;
therefore the implications also apply to procedures promised to isolate every
satisfiable circuit. They do not resolve the separate odd-count isolation
question.

The generic conversion from local comparisons to polynomial nonuniform advice
is established prior work. Ko~\cite{Ko83} gives the foundational domination
method. Naik, Rogers, Royer, and Selman~\cite[Section 2, Lemma 4]{NRRS98}
use polynomial families of small advice sets together with their membership
witnesses, attribute the multiway domination lemma to Ogihara~\cite{Ogihara96},
and describe the greedy averaging-and-deletion principle. The present proof
rederives that conversion for the tagged promise tester, including the exact
sampling bijection and the resource bounds when the number of regions grows.
The contribution is the tagged affine realization and the multiscale
isolation consequence, not a new generic selector-domination theorem.

Membership comparability supplies related but distinct hypotheses.
Ogihara~\cite{Ogihara95} studies algorithms excluding one possible vector of
membership bits; Sivakumar~\cite{Sivakumar99} shows that logarithmic membership
comparability of SAT yields a polynomial-time algorithm for the
unique-versus-zero promise. The present argument constructs no such comparator
and does not import either result as a collapse theorem. In particular, it
does not rely on the general polynomial-arity membership-comparability circuit
claim whose retraction is reported by Beigel, Fortnow, and Pavan~\cite[Section
1.1]{BFP}. All promise-transfer and domination steps used here are proved
locally.

The multiscale estimate is the affine-flag argument recorded in~\cite{BD26}.
Its new use here is within a canonical, syntax-dependent tagged-gadget
comparison: unlike presentation-invariant isolation, the argument never
identifies the output laws of different equivalent circuit descriptions.

\clearpage
\phantomsection

\end{document}